\documentclass[12pt]{amsart}
\usepackage{amsmath, amsthm, amsfonts, amsbsy, amssymb, upref, enumerate, bigstrut,  color, mathtools, mathrsfs, float, bm, dsfont,scalefnt}

\usepackage[left=1.3in,top=1.3in,right=1.3in]{geometry}

\usepackage{epsfig, graphicx}
\usepackage{ hyperref}

\newtheorem{theorem}{Theorem}[section]
\newtheorem{lemma}[theorem]{Lemma}
\newtheorem{corollary}[theorem]{Corollary}
\newtheorem{proposition}[theorem]{Proposition}
\newtheorem{definition}[theorem]{Definition}

\theoremstyle{definition}
\newtheorem{remark1}{Remark}

\newtheorem{thm1}{Theorem} 

\begin{document}
\title[Sharp Phase Transition for the RCM]
{{Sharp Phase Transition for the Random-Cluster
		Model with Summable External Magnetic Field 
}}
%
%
%
\author{R. Vila}
\address{
\newline 
Departamento de Estat\'istica -
Universidade de Bras\'ilia, Brazil,
Email: \textup{\tt rovig161@gmail.com}
}

\date{\today}


\begin{abstract}
In this paper, we prove sharpness of the phase transition for the random-cluster model in summable positive external fields, with cluster weight $q\in\{2,3,\ldots\}$,
on the hypercubic lattice
$\mathbb{Z}^d$, $d\geqslant 2$.
That is, there exists some critical parameter $0<\beta_c<\infty$ that depends on the cluster weight and the external field, below which the model exhibits exponential decay and above which there exists almost surely an infinite cluster.
\end{abstract}

\maketitle
\section{Introduction}
The random-cluster model is defined as a correlated bond percolation model for phase transitions in lattice systems. This model was introduced by Fortuin
and Kasteleyn  around 1970 \cite{Fortuin72,FortuinCM72,FK72} satisfying specific series and parallel laws. The random-cluster model is 
a graphical representation of a range of important models of statistical mechanics, among them the independent bond percolation \cite{Dum15}, the Potts model \cite{Potts52} and the uniform spanning trees \cite{LP17}.
Graphical representations are commonly used to study the existence of
phase transitions in Ising/Potts-type models 
relating them to the random-cluster model 
through the Edwards-Sokal measure \cite{ES}.
Recently, a variety of results concerning the phase transition of the 
planar random-cluster model have emerged; see \cite{beffara-duminil,DLM15,DST16,DGHMT16}.
%
%
When non-translation-invariant external fields 
are considered, phase transition in the random-cluster model 
does not always correspond to a phase transition in the corresponding spin system.  In this case, as noticed by \cite{BBCK00},
the FKG property is harder to obtain because this property does not even hold for general boundary conditions. Moreover, other complications appear, for example, in the analysis of Dobrushin-like states
\cite{GG02} and the effect of weak boundary conditions in the Potts model \cite{BC96}.
Graphical representations \cite{Dum15} for the Potts model in external fields
have been the object of revived interest after rigorous results were proved; see \cite{BBCK00,CV16,CMR98}. 

This paper is motivated by some recent works on ferromagnetics spin systems in non-uniform external fields; see  
\cite{AGB07,BC10,BCCP15,BEv17,CV16,JS99,NOZ99,Navarrete}.
The objective of this paper is to prove that
the phase transition for the random-cluster model over the hypercubic lattice, allowing a summable positive external field, is sharp.
Roughly speaking, by using 
the strong FKG property and comparison inequalities, 
we showed that
there is a certain non-trivial critical point, $\beta_c$ above which there is almost surely a path of connected points of infinite length through the network. Below it, the model exhibits exponential decay of the infinite-volume connectivity function 
for two points.
%
Sufficient conditions for the (uniform) exponential decay of finite-volume connectivities in the planar random-cluster model, with external field, was established in \cite{Alexander04}.
In the absence of an external magnetic field, sharp phase transitions for the random-cluster model, Potts model and related percolation models, on several different types of graphs, have been studied by  \cite{MR894398,Aiz87,HM16,MR3477351,DRV19,MR18}.


This paper is organized as follows. In Section \ref{The model and the main results}, we present 
the random-cluster model with free and max-wired boundary conditions, and some definitions related to this model.
We close this section by stating the main result of paper.
%
%
%
In Section \ref{Monotonic measures at finite-volume}, we present some preliminary monotonicities of random-cluster measures. 
Finally, in Section \ref{Proof of Theorems}, we present the proof of the main result in details.

%
%
%
%
\section{The model and the main result
}
\label{The model and the main results}

In this section, we introduce the random-cluster model and state
the main result of the paper. For general results and historical background of the random-cluster model, we refer the reader to \cite{BBCK00,CV16,Dum15,G95} and the references therein.

The model is defined as follows. 
Consider a finite subgraph $G=(V,E)$ of the
infinite countable connected graph 
$\mathbb{G}= (\mathbb{Z}^d, \mathbb{E}^d)$, where $\mathbb{E}^d$ 
is the set of the nearest-neighbor edges $xy$ in the $d$-dimensional hypercubic lattice
$\mathbb{Z}^d$, with $d\geqslant 2$. 
The boundaries of $V$ and $E$ are given by
$\partial V=\{x\notin V: xy\in\mathbb{E}^d \ \text{for some} \ y\in V\}$
and
$\partial E=\{xy\in\mathbb{E}^d: xy\cap V\neq \emptyset \ \text{and} \ 
xy\cap \partial V\neq \emptyset\}$, respectively.  
The configuration space of the random-cluster model is the product space $\{0,1\}^{E\cup\partial E}$.
A generic element of this space, denoted by $\omega=(\omega_{xy})_{{xy}\in E\cup\partial E}$, is often called edge configuration.
We say that an edge $xy$ is open in a configuration $\omega$ if $\omega_{xy}=1$, and closed otherwise.  
Given $\omega\in\{0,1\}^{E\cup\partial E}$,
we set $x\longleftrightarrow y$ in $\omega$ if $x$ and $y$ are in the same connected component of $\omega$. The random sets  
$\eta(\omega)=\{xy\in E\cup\partial E:  \omega_{xy}=1\}$ and 
$C_x(\omega)=\{y\in V\cup\partial V:x\longleftrightarrow y \ \text{in} \ \omega\}\cup\{x\}$
denote the set of open edges and the open cluster 
of a vertex $x\in V$ in $\omega$, respectively. 
We write $\vert C_x\vert$ to denote the number of vertices in $C_x$.

We fix two families $\boldsymbol{J}=(J_{xy})_{xy\in \mathbb{E}^d}\in [0,\infty)^{\mathbb{E}^d}$ 
and $\boldsymbol{\widehat{h}}= (h_{x,m})_{x\in {\mathbb{Z}^d},
	m\in\{1,\ldots,q\}}\in(\mathbb{R}^{\mathbb{Z}^d})^q$
called coupling constant and external field, respectively, where $q\in\{2,3,\ldots\}$.
Throughout this paper, we define
$   
h_{x,\mathrm{ max}}
=
\max_{m\in\{1,\ldots,q\}}h_{x,m},
$ 
for each vertex $x\in V$,  denote 
\begin{align}\label{quimax}
{Q}_{x,\mathrm{ max}}(\boldsymbol{\widehat{h}})
=
\big\{m\in\{1,\ldots,q\}:h_{x,m}=h_{x,\mathrm{max}} \big\},
\end{align}
and assume that
\begin{align}\label{condicao-no-campo}
\left|{\cap}_{x\in{V}} {Q}_{x,\mathrm{ max}}(\boldsymbol{\widehat{h}})\right|\geqslant 1.
\end{align}
%
%
Fix $\beta>0$ and ${\rm m}\in\{1,\ldots,q\}$. For $\#\in \{{\rm f}, {\rm m} \}$, following the references \cite{BBCK00,CV16},
let 
$\phi^{\#}_{G;\beta,q,\boldsymbol{\widehat{h}}}$ be the measure 
satisfying, for any $\omega\in\{0,1\}^{E\cup\partial E}$,
%
%
\begin{equation}\label{rcm-def}
\phi^{\#}_{G;\beta,q,\boldsymbol{\widehat{h}}}(\omega)
=
{1\over Z^{\#}_{G;\beta,q,\boldsymbol{\widehat{h}}}}\,
\prod_{xy\in E} \big[\exp({q\beta J_{xy}})-1\big]^{\omega_{xy}}\,
\prod_{C(\omega)} 
\Theta^{\#}_{G;\beta,q,\boldsymbol{\widehat{h}}}\big[C(\omega)\big],
\end{equation}
where $Z^{\#}_{G;\beta,q,\boldsymbol{\widehat{h}}}$ is a normalizing constant,
the second product is over all open clusters of sites, and the weights of the clusters are given by
\begin{equation*}
\Theta^{\rm f}_{G;\beta,q,\boldsymbol{\widehat{h}}}\big[C(\omega)\big]
=
\sum_{m=1}^q 
\exp
\Bigg(
\beta\sum_{x\in C(\omega)}h_{x,m}
\Bigg),
\end{equation*}
for all open clusters $C(\omega)$ 
of the graph $\big(V,\eta(\omega)\big)$ with  $\omega\in\{0,1\}^{E}$ (since $\omega_{xy}=0$ for all $xy\in \partial E$); and
\begin{align*}
\Theta^{\rm m}_{G;\beta,q,\boldsymbol{\widehat{h}}}\big[C(\omega)\big]
= 
\begin{cases}
\Theta^{\rm f}_{G;\beta,q,\boldsymbol{\widehat{h}}}\big[C(\omega)\big],
& \text{if} \ C(\omega)\cap \partial V= \emptyset,
\\[0,4cm]
\displaystyle 
\exp
\Bigg(
\beta\sum_{x\in C(\omega)}h_{x,{\rm m}}
\Bigg),
& \text{otherwise},
\end{cases}
\end{align*}
for all open clusters $C(\omega)$ of the graph 
$\big(V\cup\partial V,\eta(\omega)\big)$ with $\omega\in\{0,1\}^{E\cup\partial E}$, where
$\omega_{xy}=1$ for all $xy\in \partial E$.
If ${\rm m}, \widetilde{{\rm m}}\in {\cap}_{x\in{V}} {Q}_{x,\mathrm{ max}}(\boldsymbol{\widehat{h}}),$ then 
\begin{align*}
\Theta^{\rm m}_{G;\beta,q,\boldsymbol{\widehat{h}}}\big[C(\omega)\big]
=
\Theta^{\widetilde{{\rm m}}}_{G;\beta,q,\boldsymbol{\widehat{h}}}\big[C(\omega)\big]
\quad \text{for all} \ \omega\in \{0,1\}^{E\cup\partial E},
\end{align*}
and therefore
$
\phi^{\rm m}_{G;\beta,q,\boldsymbol{\widehat{h}}}
=
\phi^{\widetilde{{\rm m}}}_{G;\beta,q,\boldsymbol{\widehat{h}}}.
$
This measure is denoted by 
$\phi^{\rm w}_{G;\beta,q,\boldsymbol{\widehat{h}}}$.
The measures $\phi^{\rm f}_{G;\beta,q,\boldsymbol{\widehat{h}}}$ and $\phi^{\rm w}_{G;\beta,q,\boldsymbol{\widehat{h}}}$ are called the random-cluster measures on $G$ with free and max-wired boundary conditions, respectively. 
By taking the weak limit of measures defined in finite volume,  the measures $\phi^{\#}_{G;\beta,q,\boldsymbol{\widehat{h}}}$, with $\#\in \{{\rm f}, {\rm w} \}$, can be extended to $\mathbb{G}$; see Theorem \ref{conv}.
%
%
The expected value of a real-valued function $f:\{0,1\}^{E\cup\partial E}\to\mathbb{R}$ 
in the random-cluster model is denoted by 
$\phi^{\#}_{G;\beta,q,\boldsymbol{\widehat{h}}} (f)$.
When $f=\mathds{1}_A$, for some cylindrical event $A$, we write $\phi^{\#}_{G;\beta,q,\boldsymbol{\widehat{h}}} (A)$ instead $\phi^{\#}_{G;\beta,q,\boldsymbol{\widehat{h}}} (\mathds{1}_A).$

The functions of principal interest are the following:
\begin{align}\label{percolation-probability}
\theta(\beta,q,\boldsymbol{\widehat{h}})
\coloneqq
\sup_{x\in\mathbb{Z}^d}
\phi^{\rm f}_{\mathbb{G};\beta,q,\boldsymbol{\widehat{h}}} 
\big(\vert C_x\vert=\infty\big);
\quad 
\chi(\beta,q,\boldsymbol{\widehat{h}})
\coloneqq
\sup_{x\in\mathbb{Z}^d}
\phi^{\rm w}_{\mathbb{G};\beta,q,\boldsymbol{\widehat{h}}}\big(\vert C_x\vert\big).
\end{align}
These functions are known as percolation-probability and susceptibility function, respectively. 

We write $\boldsymbol{\widehat{h}}\in\ell^{1}(\mathbb{Z}^d)$ when $\boldsymbol{\widehat{h}}$ is summable, that is,  $\lVert\boldsymbol{\widehat{h}}\rVert_1=\sum_{m=1}^{q}\sum_{x\in \mathbb{Z}^d} \vert h_{x,m}\vert<\infty$.
The following theorem is the main result of paper.
This result discusses 
the existence of phase transition 
and the exponential decay of the infinite-volume connectivity function 
for two points.

\begin{thm1}\label{phase-transition-existence}
	Fix $q\in\{2,3,\ldots\}$ and $d\geqslant 2$. Consider
	the random-cluster model on $\mathbb{G}$, with uniform coupling constant  $\boldsymbol{J}$, i.e. $J_{xy}=J\geqslant 0$ is constant for all $xy\in\mathbb{E}^d$,  and external field $\boldsymbol{\widehat{h}}= (h_{x,m})_{x\in \mathbb{Z}^d, m\in\{1,\ldots,q\}}\in(\mathbb{R}^{\mathbb{Z}^d})^q$ satisfying the condition \eqref{condicao-no-campo}.
	Then, there exists a critical point $0<\beta_c=\beta_c(q,\boldsymbol{\widehat{h}})<\infty$ so that
	\begin{itemize}
		\item[\rm 1)] 	
		$\theta(\beta,q,\boldsymbol{\widehat{h}})=0$ if $\beta<\beta_c$, and $\theta(\beta,q,\boldsymbol{\widehat{h}})>0$ if $\beta>\beta_c$.
		\item[\rm 2)]
		$\chi(\beta,q,\boldsymbol{\widehat{h}})<\infty$ if $\beta<\beta_0,$
		for some 
		$0<\beta_0<\beta_c$; and 
		$\chi(\beta,q,\boldsymbol{\widehat{h}})=\infty$ if $\beta>\beta_c$.
		\item[\rm 3)]  
		In the case that $\boldsymbol{\widehat{h}}\in\ell^{1}(\mathbb{Z}^d)$ has positive terms,
		for all $\beta<\beta_c$
		there exist constants $C(\beta,\boldsymbol{\widehat{h}}),\gamma(\beta,q)>0$, such that
		\begin{align*}
		\displaystyle
		\phi^{\rm w}_{\mathbb{G};\beta,q,\boldsymbol{\widehat{h}}}
		(x\longleftrightarrow y)
		\leqslant 
		C(\beta,\boldsymbol{\widehat{h}}) \,
		{\exp}\big[{-\|x-y\| \gamma(\beta,q)}\big]
		\quad \text{for all} \ x,y\in\mathbb{Z}^d,
		\end{align*}	
		where $\|\cdot\|$ denotes the Euclidean norm.
		This implies finite susceptibility.
	\end{itemize}
\end{thm1}

This theorem extends to a large class of amenable infinite 
transitive graphs.

As a consequence of Theorem \ref{phase-transition-existence} we have the sharpness of the phase transition for the random-cluster model in summable positive external fields,  on the hypercubic lattice. Remains to investigate the following. For which classes of non-summable external fields, the phase transition of the random-cluster on the graph $\mathbb{G}$ is sharp? 
%
%
Sharp phase transitions for spin systems and related percolation models, in the absence of an external magnetic field, on different types of graphs have emerged in references  \cite{MR894398,Aiz87,HM16,MR3477351,DRV19,MR18}.

The rest of this paper is devoted to the proof of  Theorem  \ref{phase-transition-existence}. To this end, some comparison inequalities are essential.

\section{Monotonic measures}
\label{Monotonic measures at finite-volume}
Let $G=(V,E)$ be a finite subgraph of $\mathbb{G}$ and $q\in\{2,3,\ldots\}$.
We consider the usual partial order on $\{0,1\}^{E\cup\partial E}$ where
$
\omega\preceq \tilde{\omega} 
$ iff $\omega_{xy}\leqslant \tilde{\omega}_{xy}$ for all $xy\in E\cup\partial E$.

\begin{definition}
	Let $(\Omega,\preceq)$ be a partially ordered space. 
	A measure $\mu$ over $\Omega$ said to have the  
	$\mathrm{FKG}$ property if 
	\[
	\mu(fg)\geqslant \mu(f)\mu(g),
	\]
	for any increasing (with respect to $\preceq$) 
	measurable functions $f,g:\Omega\to \mathbb{R}$. 
	Furthermore, if $\Omega$ is a cartesian product 
	$\Omega=\prod_{xy\in B}\Omega_{xy}$, with $|\Omega_{xy}|<\infty$,
	then $\mu$ is said to have 
	the {\bf strong FKG property}, if 
	$\mu(\cdot\vert A)$ has the $\mathrm{FKG}$ property for each 
	cylinder event  
	$
	A
	=
	\{
	\omega\in\Omega:\omega_{xy}=\alpha_{xy} \ \text{for all} \ xy\in\{0,1\}^{E\cup\partial E}\in \widetilde{B}
	\}
	$,
	where $\widetilde{B}\subset{B}$ is finite and 
	$\alpha_{xy}\in\Omega_{xy}$ for all $xy\in\widetilde{B}$.
\end{definition}
\begin{theorem}[\cite{CV16}] \label{Strong FKG Property}
	The finite-volume measures
	$\phi^{\rm f}_{G;\beta,q,\boldsymbol{\widehat{h}}}$ 
	and 
	$\phi^{\rm w}_{G;\beta,q,\boldsymbol{\widehat{h}}}$
	have the strong $\mathrm{FKG}$ property.
\end{theorem}

By using the strong FKG property for the random-cluster model, one can prove the following two theorems.
\begin{theorem}[\cite{CV16}]\label{conv}
	For each increasing quasilocal function f (see \cite{Georgii88}),
	\begin{itemize}
		\item[\rm (a)]  The following limits exist
		\begin{align*}
		\phi^{\rm f}_{\mathbb{G};\beta,q,\boldsymbol{\widehat{h}}} (f)
		=
		\lim_{G\uparrow \mathbb{G}}
		\phi^{\rm f}_{G;\beta,q,\boldsymbol{\widehat{h}}} (f);
		\quad 
		\phi^{\rm w}_{\mathbb{G};\beta,q,\boldsymbol{\widehat{h}}} (f)
		=
		\lim_{G\uparrow \mathbb{G}}
		\phi^{\rm w}_{G;\beta,q,\boldsymbol{\widehat{h}}} (f).
		\end{align*}
		\item[\rm (b)] 
		The finite-volume measures $\phi^{\rm f}_{{G};\beta,q,\boldsymbol{\widehat{h}}}$ and $\phi^{\rm w}_{{G};\beta,q,\boldsymbol{\widehat{h}}}$ are  the extremal ones,
		in the sense that, if 
		$\phi_{{G};\beta,q,\boldsymbol{\widehat{h}}}$ is a finite-volume random-cluster measure, then
		\begin{align*}
		\phi^{\rm f}_{{G};\beta,q,\boldsymbol{\widehat{h}}} (f)
		\leqslant
		\phi_{{G};\beta,q,\boldsymbol{\widehat{h}}} (f)
		\leqslant
		\phi^{\rm w}_{{G};\beta,q,\boldsymbol{\widehat{h}}} (f).
		\end{align*}
		Consequently,	
		$\phi^{\rm f}_{\mathbb{G};\beta,q,\boldsymbol{\widehat{h}}}$ and $\phi^{\rm w}_{\mathbb{G};\beta,q,\boldsymbol{\widehat{h}}}$ are also extremal measures.
	\end{itemize}
\end{theorem}

Hereafter  $N_{\infty}$ denotes the random variable that counts the number of infinite open clusters in the sample space $\{0,1\}^{\mathbb{E}^d}$.
The following theorem states that the probability that a unique infinite cluster exists is either zero or one.
\begin{theorem}[\cite{CV16}]\label{unique-cluster}
	The measures 
	$\phi^{\#}_{\mathbb{G};\beta,q,\boldsymbol{\widehat{h}}}$, with $\#\in \{{\rm f}, {\rm w} \}$,  have the almost-sure uniqueness of the infinite open cluster property. That is, 
	\begin{align*}
	\phi^{\#}_{\mathbb{G};\beta,q,\boldsymbol{\widehat{h}}} (N_{\infty}\leqslant 1)=1
	\quad \text{for} \ \#\in \{{\rm f}, {\rm w} \}.
	\end{align*}
\end{theorem}
The next monotonicity result is valid for a large class of probability measures with 
the restriction that these measures are positively associated.
\begin{proposition}\label{proposition-basic}
	Let $(\Omega,\mathscr{F}, \mu_i)$, $i=1,2,$  two probability spaces with probability measures
	$
	\mu_i(\omega)
	=
	\mathcal{K}(\omega)
	\mathcal{W}_i(\omega)/Z_i
	$, $\omega\in\Omega$, $i=1,2,$ satisfying the FKG property, where
	$Z_i$ is the respective partition function and $\mathcal{K}$ is a non-negative function. 
	%
	If $\mathcal{W}_1/\mathcal{W}_2$ is an increasing non-negative function, 
	then $\mu_1(f)\geqslant \mu_2(f)$ for each cylindrical increasing function $f$.	
	%
\end{proposition}
\begin{proof}
	Let
	$g(\omega)= {\mathcal{W}_1(\omega)/\mathcal{W}_2(\omega)}$, $\omega\in\Omega$.  Note that $g$ is an increasing function.
	A simple computation shows that, for all cylindrical valued-real function $f$,
	\begin{align*}
	\mu_1(f)
	=
	{Z_2\over Z_1}\, \mu_2(fg).
	\end{align*}
	Taking $f\equiv 1$ we obtain that
	$
	\mu_2(g)={Z_1/Z_2}.
	$		
	Then, by using the above identity, from FKG inequality we have, for each cylindrical increasing function $f$,
	\[ 
	\mu_1(f)={\mu_2(f g)\over \mu_2(g)}\geqslant \mu_2(f).
	\]
	Thus the proof is complete.
	%
	%
\end{proof}

The next lemma shows that the part of the measures $\phi^{\#}_{G;\beta,q,\boldsymbol{\widehat{h}}}$ in \eqref{rcm-def}, {with} $\#\in \{{\rm f}, {\rm w}\}$, due to the external magnetic field are decreasing functions.
%
\begin{lemma}\label{proposition-monotonicity} 
	Under the condition \eqref{condicao-no-campo}, the following hold:
	\begin{itemize}
		\item[\rm (a)]
		The function ${g}^{\rm f}$ defined as, for any $\omega\in\{0,1\}^{E}$, 
		\[
		{g}^{\rm f}(\omega)
		\coloneqq
		\prod_{C(\omega)} 	
		\Theta^{\rm f}_{G;\beta,q,\boldsymbol{\widehat{h}}}\big[C(\omega)\big]
		=
		\prod_{C(\omega)} 
		\sum_{m=1}^q 
		\exp
		\Bigg(
		\beta\sum_{x\in C(\omega)}h_{x,m}
		\Bigg),
		\]
		is monotone decreasing in the FKG sense, where the above product runs over all the open clusters $C(\omega)$ 
		of the graph $\big(V,\eta(\omega)\big)$.
		\item[\rm (b)]
		The function ${g}^{\rm w}$ defined as, for any $\omega\in\{0,1\}^{E\cup\partial E}$, 
		\begin{align*}
		{g}^{\rm w}(\omega)
		&\coloneqq
		\prod_{C(\omega)} 	
		\Theta^{\rm w}_{G;\beta,q,\boldsymbol{\widehat{h}}}\big[C(\omega)\big]
		\\
		&=
		\prod_{C(\omega):C(\omega)\cap\partial V=\emptyset }  \,
		\sum_{m=1}^q 
		\exp
		\Bigg(
		\beta\sum_{x\in C(\omega)}h_{x,m}
		\Bigg)
		\prod_{C(\omega):C(\omega)\cap\partial V\neq\emptyset } 	
		\exp
		\Bigg(
		\beta\sum_{x\in C(\omega)}h_{x,{\rm max}}
		\Bigg),
		\end{align*}
		is monotone decreasing in the FKG sense.
	\end{itemize}
\end{lemma}
\begin{proof}	
	Let $xy$ be a nearest-neighbor edge such that
	$\omega_{xy}=0$ and let $\omega^{xy}$ be the configuration obtained by flipping $\omega_{xy}$ to $1$.		
	It suffices to study the single-bond flips. That is,
	it is sufficient to prove the following 
	\begin{align}\label{condicao-suficiente-1}
	{g}^{\#}(\omega^{xy})\leqslant {g}^{\#}(\omega)
	\quad \text{with} \ \#\in \{{\rm f}, {\rm w} \}.
	\end{align}

	If the vertices $x$ and $y$ are connected in the configuration $\omega$, ${g}^{\#}(\omega^{xy})={g}^{\#}(\omega)$. 
	If the vertices $x$ e $y$ are not connected in  $\omega$,
	then there are two open clusters 
	$A= C_x(\omega)$ and $B= C_y(\omega)$ containing the vertices
	$x$ and $y$, respectively. 
	If $xy$ is an open edge at $\omega$, then the clusters $A$ and $B$ are connected, 
	creating a new open cluster denoted by
	$C= A\cup B$. Then, $|C|=|A|+|B|,$ and so to prove  
	\eqref{condicao-suficiente-1}, it is enough to verify that 
	\begin{align}\label{cond-show}
	\Theta^{\#}_{G;\beta,q,\boldsymbol{\widehat{h}}}\big[C\big]
	\leqslant
	\Theta^{\#}_{G;\beta,q,\boldsymbol{\widehat{h}}}\big[A\big]
	\Theta^{\#}_{G;\beta,q,\boldsymbol{\widehat{h}}}\big[B\big]
	\quad \text{with} \ \#\in \{{\rm f}, {\rm w} \}.
	\end{align}	
	
	(a) When $\#={\rm f}$, the inequality in \eqref{cond-show} can be rewritten as
	\begin{align*}
	\sum_{m=1}^q 
	\exp
	\Bigg(
	\beta\sum_{x\in C}h_{x,m}
	\Bigg)
	\leqslant
	\sum_{m=1}^q 
	\exp
	\Bigg(
	\beta\sum_{x\in A}h_{x,m}
	\Bigg)
	\sum_{m=1}^q 
	\exp
	\Bigg(
	\beta\sum_{x\in B}h_{x,m}
	\Bigg).
	\end{align*}
	But the above inequality is true since 
	\begin{multline*}
	\sum_{k=1}^{q} \exp
	\Bigg(
	\beta\sum_{x\in A}h_{x,k}
	\Bigg)
	\sum_{m=1}^{q} \exp
	\Bigg(
	\beta\sum_{x\in B}h_{x,m}
	\Bigg)
	\\[0,2cm]
	\geqslant
	\sum_{k=1}^{q} \exp
	\Bigg(
	\beta\sum_{x\in A}h_{x,k}
	\Bigg) 
	\exp
	\Bigg(
	\beta\sum_{x\in B}h_{x,\mathrm{max}}
	\Bigg)
	\left|{\cap}_{x\in{V}} {Q}_{x,\mathrm{ max}}(\boldsymbol{\widehat{h}})\right|
	\\
	\stackrel{}{\geqslant}
	\sum_{j=1}^{q} 
	\exp
	\Bigg(
	\beta\sum_{x\in C}h_{x,j}
	\Bigg),
	\end{multline*}
	where in the last line we have used the condition \eqref{condicao-no-campo}.
	
	(b) When $\#={\rm w}$, we analyze the following possible cases.
	(i) $A\cap\partial V=\emptyset$ and $B\cap\partial V=\emptyset$; 
	(ii) $A\cap\partial V\neq\emptyset$ and $B\cap\partial V=\emptyset$;
	(iii) $A\cap\partial V=\emptyset$ and $B\cap\partial V\neq\emptyset$; and
	(iv) $A\cap\partial V\neq\emptyset$ and $B\cap\partial V\neq\emptyset$.
	
	In the case (i), the inequality in \eqref{cond-show} follows directly from Item (a). In the case (ii), \eqref{cond-show} is equivalent to the following inequality
	\begin{align*}
	\exp
	\Bigg(
	\beta\sum_{x\in C}h_{x,{\rm max}}
	\Bigg)
	\leqslant
	\exp
	\Bigg(
	\beta\sum_{x\in A}h_{x,{\rm max}}
	\Bigg)
	\sum_{m=1}^q 
	\exp
	\Bigg(
	\beta\sum_{x\in B}h_{x,m}
	\Bigg).
	\end{align*}
	But this is true since 
	\begin{multline*}
	\exp
	\Bigg(
	\beta\sum_{x\in A}h_{x,{\rm max}}
	\Bigg)
	\sum_{m=1}^q 
	\exp
	\Bigg(
	\beta\sum_{x\in B}h_{x,m}
	\Bigg)
	\\[0,1cm]
	\geqslant
	\exp
	\Bigg(
	\beta\sum_{x\in A}h_{x,{\rm max}}
	\Bigg)
	\exp
	\Bigg(
	\beta\sum_{x\in B}h_{x,\mathrm{max}}
	\Bigg)
	\left|{\cap}_{x\in{V}} {Q}_{x,\mathrm{ max}}(\boldsymbol{\widehat{h}})\right|
	\\
	\geqslant
	\exp
	\Bigg(
	\beta\sum_{x\in C}h_{x,{\rm max}}
	\Bigg),
	\end{multline*}	
	where, again, in the last line we have used the condition \eqref{condicao-no-campo}.
	In the case (iii), the inequality in \eqref{cond-show} is proven analogously to Item (ii). Finally, in the case (iv), \eqref{cond-show} is equivalent to the following identity 
	\begin{align*}
	\exp
	\Bigg(
	\beta\sum_{x\in C}h_{x,{\rm max}}
	\Bigg)
	=
	\exp
	\Bigg(
	\beta\sum_{x\in A}h_{x,{\rm max}}
	\Bigg)
	\exp
	\Bigg(
	\beta\sum_{x\in B}h_{x,{\rm max}}
	\Bigg).
	\end{align*}
	Hence the proof follows.
\end{proof}	
%
%

%
It is known that stochastic dominance allows to compare different probability measures on the same sample space. 
One of the common uses of this type of comparisons is to prove that certain 
characteristics are inherited from one random model to another, which is generally not trivial. 
The following lemma states that, at finite volume, the Bernoulli measure  dominates 
the random-cluster measure with external field. 
Hence, for example, we can use this lemma to prove that,
at the thermodynamic limit, the almost-sure uniqueness (under amenabilitity condition in the infinite graph) of the infinite open cluster would be inherited 
for the infinite-volume random-cluster measure; see Theorem \ref{unique-cluster} and Theorem 10 in \cite{CV16} for a rigorous proof of this fact.
\begin{lemma}
	\label{Stochastic domination of Bernoulli}
	For each cylindrical increasing function $f$ we have
	\[
	\phi^{\#}_{G;\beta,q,\boldsymbol{\widehat{h}}} (f)
	\leqslant
	\phi^{\#}_{G;\beta,1,\boldsymbol{\widehat{h}}} (f)
	=
	\phi^{\#}_{G;\beta,1,\boldsymbol{\widehat{0}}} (f)
	\eqqcolon 
	\mathbb{P}_{G;\beta}(f)
	\quad \text{with} \ \#\in \{{\rm f}, {\rm w} \}, 
	\]
	where $\mathbb{P}_{G;\beta}$ becomes 
	the product probability measure on $\{0,1\}^{E}$ defined by the Bernoulli factor 
	$\prod_{xy\in E} p_{xy}^{\omega_{xy}}(1-p_{xy})^{1-\omega_{xy}}$,
	with $p_{xy}=1-\exp({-q\beta J_{xy}})$.
\end{lemma}
\begin{proof}
	Let $\mu_1$ and $\mu_2$ be two probability measures defined by $\mathbb{P}_{G;\beta}$ and
	$\phi^{\#}_{G;\beta,q,\boldsymbol{\widehat{h}}}$, 
	respectively. The respective weights of these measures are given by
	$\mathcal{W}_1(\omega)= \exp\big({-q\beta \sum_{xy\in E}J_{xy}}\big)$,
	$\mathcal{W}_2(\omega)= {g}^{\#}(\omega)$ and $\mathcal{K}(\omega)=\prod_{xy\in E} \big[\exp({q\beta J_{xy}})-1\big]^{\omega_{xy}}$,
	$\omega\in\{0,1\}^{E}$, 
	where ${g}^{\rm f}$ and ${g}^{\rm w}$ are given in Lemma \ref{proposition-monotonicity}.
	Since $\mathcal{W}_2$ 
	is decreasing (see Lemma \ref{proposition-monotonicity}),  $\mathcal{W}_1/\mathcal{W}_2$ is an increasing function. Hence,
	by Proposition \ref{proposition-basic},  we have that
	$\mu_1(f)\geqslant \mu_2(f)$ 
	%
	for each cylindrical increasing function $f$.
	
	On the other hand, 	the identity $	\phi^{\#}_{G;\beta,1,\boldsymbol{\widehat{h}}} (f)
	= 
	\mathbb{P}_{G;\beta}(f)$ follows immediately by taking $q=1$ in \eqref{rcm-def}.
\end{proof}

In what follows we consider as usual the partial order on $[0,\infty)^{E}$ where
$
\boldsymbol{J}
\curlyeqprec 
\boldsymbol{J}'$ iff $J_{xy}\leqslant J_{xy}'$
for all $xy\in E.$
\begin{proposition}\label{proposition-monotonicity-1} 
	Suppose that $\boldsymbol{J}\curlyeqprec  \boldsymbol{J}'$ are two coupling constants. The function $\widetilde{g}$ defined as, for any $\omega\in\{0,1\}^{E}$, 
	\[
	\widetilde{g}(\omega)
	\coloneqq
	{\prod_{xy\in E} \big[\exp({q\beta J_{xy}})-1\big]^{\omega_{xy}} \over 
		\prod_{xy\in E} \big[\exp({q\beta J_{xy}'})-1\big]^{\omega_{xy}}}, 
	\]
	is monotone decreasing in the FKG sense.
\end{proposition}
\begin{proof} Let $x_0y_0$ be a nearest-neighbor edge such that
	$\omega_{x_0y_0}=0$ and let $\omega^{x_0y_0}$ be the configuration obtained by flipping $\omega_{x_0y_0}$ to $1$.
	To prove the monotonicity of $\widetilde{g}$ it is sufficient to prove
	that $\widetilde{g}(\omega^{x_0y_0})\leqslant \widetilde{g}(\omega)$.
	But this is true because	
	\begin{align*}
	{\widetilde{g}(\omega^{x_0y_0})\over\widetilde{g}(\omega) }
	&= 
	{
		\prod_{xy\in E}\big[{\exp}({q\beta J_{xy}})-1\big]^{\omega^{x_0y_0}_{xy}} 
		\over 
		\prod_{xy\in E}\big[{\exp}({q\beta J_{xy}'})-1\big]^{\omega^{x_0y_0}_{xy}} 
	}
	\,
	{
		\prod_{xy\in E}\big[{\exp}({q\beta J_{xy}'})-1\big]^{\omega_{xy}} 
		\over 
		\prod_{xy\in E}\big[{\exp}({q\beta J_{xy}})-1\big]^{\omega_{xy}} 
	}
	\\[0,15cm]
	&=
	{
		{\exp}({q\beta J_{x_0y_0}})-1
		\over 
		{\exp}({q\beta J_{x_0y_0}'})-1
	}
	\leqslant
	1,
	\end{align*}
	whenever $\boldsymbol{J}\curlyeqprec  \boldsymbol{J}'$. Then the proof follows.
\end{proof}	

Our next result is the monotonicity, in the FKG 
sense, with respect to the coupling constants. 	The first item of this result is proved in reference \cite{CV16}  when the coupling constant is  uniform.
\begin{proposition}\label{Monotonicity with respect to the coupling}
	\label{teorema ge}
	Denote by
	$\phi^{{\#},\boldsymbol{J}}_{G;\beta,q,\boldsymbol{\widehat{h}}}$
	the random-cluster measure on $G$ with $\#\in \{{\rm f}, {\rm w} \}$ boundary condition and
	with coupling constant 		
	$\boldsymbol{J}=
	(J_{xy})_{xy\in {E}}$ $\in [0,\infty)^{E}$.
	Suppose that $\boldsymbol{J}\curlyeqprec  \boldsymbol{J}'$ are two coupling constants. 
	Then, for any cylindrical increasing function $f$, 
	\begin{itemize}
		\item[\rm (a)] 
		$\displaystyle
		\phi^{{\#},\boldsymbol{J}}_{G;\beta,q,\boldsymbol{\widehat{h}}} (f)
		\leqslant 
		\phi^{{\#},\boldsymbol{J}'}_{G;\beta,q,\boldsymbol{\widehat{h}}} (f).
		$
		\item[\rm (b)] 
		In the particular case that $\boldsymbol{J}$ is an uniform coupling constant,
		\begin{itemize}
			\item[\rm (i)]  the mapping 
			$\beta\longmapsto \phi^{{\#}}_{G;\beta,q,\boldsymbol{\widehat{h}}}(f)$ is increasing;
			\item[\rm (ii)] the mappings $\beta\longmapsto \theta(\beta,q,\boldsymbol{\widehat{h}})$ and 
			$\beta\longmapsto \chi(\beta,q,\boldsymbol{\widehat{h}})$ are increasing.
		\end{itemize}
	\end{itemize}
\end{proposition}
\begin{proof}
	In order to prove Item (a),
	we only present the argument for $\#={\rm w}$, since for the free
	boundary condition case the proof works similarly.	
	
	Let $\mu_1$ and $\mu_2$ be two probability measures defined by $\phi^{{\rm w},\boldsymbol{J}'}_{G;\beta,q,\boldsymbol{\widehat{h}}}$ and
	$\phi^{{\rm w},\boldsymbol{J}}_{G;\beta,q,\boldsymbol{\widehat{h}}}$, 
	respectively. The respective weights of these measures are given by,
	for any $\omega\in\{0,1\}^{E\cup\partial E}$,
	\begin{align*}
	&\mathcal{W}_1(\omega)= \prod_{xy\in E} \big[\exp({q\beta J_{xy}'})-1\big]^{\omega_{xy}}\, {g}^{\rm w}(\omega);
	\\[0,15cm]
	&\mathcal{W}_2(\omega)= \prod_{xy\in E} \big[\exp({q\beta J_{xy}})-1\big]^{\omega_{xy}}\, {g}^{\rm w}(\omega);
	\end{align*}
	and $\mathcal{K}(\omega)=1$,
	where the function ${g}^{\rm w}$ is given in Lemma \ref{proposition-monotonicity}.
	Since $\boldsymbol{J}\curlyeqprec \boldsymbol{J}'$, by Proposition \ref{proposition-monotonicity-1}, the function $\mathcal{W}_2/\mathcal{W}_1$, given by 
	\begin{align*}
	{\mathcal{W}_2(\omega)\over\mathcal{W}_1(\omega)}
	=
	{\prod_{xy\in E} \big[\exp({q\beta J_{xy}})-1\big]^{\omega_{xy}} \over 
		\prod_{xy\in E} \big[\exp({q\beta J_{xy}'})-1\big]^{\omega_{xy}}}
	=\widetilde{g}(\omega),
	\end{align*}
	is monotone decreasing in the FKG sense. Therefore,  $\mathcal{W}_1/\mathcal{W}_2$ is an increasing function. Hence,
	by Proposition \ref{proposition-basic},  we have that
	$\mu_1(f)\geqslant \mu_2(f)$ 
	for each cylindrical increasing function $f$. This proves the inequality of the first item. 
	
	The statements in Item (b) follow by direct application of Item (a). So we have finished the proof.
\end{proof}
%


Since we are also interested in monotonicity properties with
respect to the magnetic field, 
it is needed to introduce a partial order between two 
fields as in reference \cite{CV16}. 
Given two arbitrary magnetic fields
$\boldsymbol{\widehat{h}}=
(h_{x,m})_{x\in{V}, m\in\{1,\ldots,q\}}$
and 
$\boldsymbol{\widehat{h}}'=
(h_{x,m}')_{x\in{V}, m\in\{1,\ldots,q\}}
$
in
$ 
(\mathbb{R}^{V})^q$,
we say that
\begin{align}\label{relação de ordem nos campos}
\boldsymbol{\widehat{h}}\prec \boldsymbol{\widehat{h}}'
\ 
\Longleftrightarrow
\ 
\forall \ x\in V: \ 
h_{x,k}-h_{x,l}\leqslant h'_{x,k}-h'_{x,l},
\ \ \
k,l=1,\ldots,q,
\end{align}
whenever $h_{x,k}-h_{x,l}>0.$
\begin{theorem}[\cite{CV16}]
	\label{teo-monotonicidade-campo-externo-grc}
	Let
	$\boldsymbol{\widehat{h}}\prec \boldsymbol{\widehat{h}}'$ be two arbitrary magnetic fields in $(\mathbb{R}^{V})^q$.
	Then, for any  quasilocal  increasing function $f$, 
	\begin{align*}
	\phi^{\#}_{G;\beta,q,\boldsymbol{\widehat{h}}} (f)
	\leqslant 
	\phi^{\#}_{G;\beta,q,\boldsymbol{\widehat{h}}'} (f)
	\quad \text{with} \ \#\in \{{\rm f}, {\rm w} \}.
	\end{align*}
\end{theorem}

The following result is a direct application of the above theorem.
\begin{corollary}\label{ineq-fund}
	If  $\boldsymbol{\widehat{h}}$ is an arbitrary external field in $(\mathbb{R}^{V})^q$,
	then, for any cylindrical increasing function $f$, 
	\begin{align*}
	\phi^{\#}_{G;\beta,q,\boldsymbol{\widehat{0}}} (f)
	\leqslant 
	\phi^{\#}_{G;\beta,q,\boldsymbol{\widehat{h}}} (f)
	\quad \text{with} \ \#\in \{{\rm f}, {\rm w} \}.
	\end{align*}
\end{corollary}	
\begin{proof}
	Given $\varepsilon>0$ small enough,
	we consider the non-uniform external magnetic field 
	$
	\varepsilon\boldsymbol{\widehat{h}}\coloneqq(
	\varepsilon h_{x,m})_{x\in{V}, m\in\{1,\ldots,q\}}
	$.
	A straightforward computation shows that 
	$\varepsilon \boldsymbol{\widehat{h}} \prec \boldsymbol{\widehat{h}}$, where the partial order ‘$\prec$’ is given in \eqref{relação de ordem nos campos}. 
	Then, by Theorem \ref{teo-monotonicidade-campo-externo-grc}, it follows 
	that 
	\begin{align*}
	\phi_{G;\beta,q,\varepsilon\boldsymbol{\widehat{h}}}^{\#} (f)
	\leqslant 
	\phi_{G;\beta,q,\boldsymbol{\widehat{h}}}^{\#} (f)
	\quad \text{with} \ \#\in \{{\rm f}, {\rm w} \},
	\end{align*}
	for each cylindrical increasing function $f$.
	Taking $\epsilon\to 0^+$ in this inequality we get that 
	$
	\phi_{G;\beta,q,\boldsymbol{\widehat{0}}}^{\#} (f)
	\leqslant 
	\phi_{G;\beta,q,\boldsymbol{\widehat{h}}}^{\#} (f).
	$	
	We thus complete the proof.
\end{proof}

The next result compares the  random-cluster measure in summable positive external fields  with the random-cluster measure 
of their counterparts under the zero external field.
\begin{lemma}\label{lemma-teo2} If
	$\boldsymbol{\widehat{h}}\in\ell^{1}(\mathbb{Z}^d)$ is a summable external field with positive terms,
	there exists $C(\beta,\boldsymbol{\widehat{h}})>0$ such that, for any non-negative cylindrical function $f$,
	\begin{align*}\
	\phi^{\#}_{G;\beta,q,\boldsymbol{\widehat{h}}} (f)
	\leqslant
	C(\beta,\boldsymbol{\widehat{h}}) \,
	\phi^{\rm f}_{G;\beta,q,\boldsymbol{\widehat{0}}} (f)
	\quad \text{with} \ \#\in \{{\rm f}, {\rm w} \}.
	\end{align*}
\end{lemma}
\begin{proof}
	We only present the argument for $\#={\rm w}$, since for the free
	boundary condition case the proof works similarly.	
	%
	Indeed,
	for all open cluster $C(\omega)$  such that
	$C(\omega)\cap\partial V=\emptyset$,
	since $h_{x,{\rm max}}>h_{x,m}$ for all  $x$ and $m$, we have
	\begin{align*}
	\sum_{m=1}^q 
	\exp
	\Bigg(
	\beta\sum_{x\in C(\omega)}h_{x,m}
	\Bigg) 
	\leqslant 
	q \boldsymbol{\cdot} \exp
	\Bigg(
	\beta\!\sum_{x\in C(\omega)}h_{x,{\rm max}}
	\Bigg).
	\end{align*}
	This implies that, for any open cluster $C(\omega)$ of the graph 
	$\big(V\cup\partial V,\eta(\omega)\big)$ with $\omega\in\{0,1\}^{E\cup\partial E}$,
	{
		\begin{align}
		&\prod_{C(\omega)} 	
		\Theta^{\rm w}_{G;\beta,q,\boldsymbol{\widehat{h}}}\big[C(\omega)\big] \nonumber
		\\[0,1cm]
		&=
		\prod_{C(\omega):C(\omega)\cap\partial V=\emptyset }  \,
		\sum_{m=1}^q 
		\exp
		\Bigg(
		\beta\sum_{x\in C(\omega)}h_{x,m}
		\Bigg)
		\prod_{C(\omega):C(\omega)\cap\partial V\neq\emptyset } 	
		\exp
		\Bigg(
		\beta\sum_{x\in C(\omega)}h_{x,{\rm max}}
		\Bigg) \nonumber
		\\[0,1cm]
		& 
		\leqslant
		q^{\big\vert\big\{C(\omega):\, C(\omega)\cap\partial V=\emptyset \big\}\big\vert}
		\,
		\exp	
		\Bigg(
		\beta\sum_{x\in V\cup\partial V}h_{x,{\rm max}}
		\Bigg).
		%
		\label{object-1}
		\end{align}
	}
	
	Since  $\boldsymbol{\widehat{h}}$ has positive terms,
	$
	Z^{\rm w}_{G;\beta,q,\boldsymbol{\widehat{h}}}
	\geqslant
	Z^{\rm f}_{G;\beta,q,\boldsymbol{\widehat{0}}} 
	\,
	\exp	
	(
	\beta\sum_{x\in \partial V}h_{x,{\rm max}}
	)
	$. Hence, by combining this inequality with the one  in \eqref{object-1},  it follows that, for any non-negative cylindrical function $f$,
	\begin{multline*}
	{1\over Z^{\rm w}_{G;\beta,q,\boldsymbol{\widehat{h}}}}\,
	\sum_{\omega}f(\omega)
	\prod_{xy\in E} \big[\exp({q\beta J_{xy}})-1\big]^{\omega_{xy}}\,
	\prod_{C(\omega)} 
	\Theta^{\rm w}_{G;\beta,q,\boldsymbol{\widehat{h}}}\big[C(\omega)\big]
	\\[0,15cm]
	\leqslant
	{
		\exp	
		\big(
		\beta\sum_{x\in V}h_{x,{\rm max}}
		\big)		
		\over Z^{\rm f}_{G;\beta,q,\boldsymbol{\widehat{0}}}} \,
	\sum_{\omega}f(\omega)
	\prod_{xy\in E} \big[\exp({q\beta J_{xy}})-1\big]^{\omega_{xy}}\,
	q^{\big\vert\big\{C(\omega):\, C(\omega)\cap\partial V=\emptyset \big\}\big\vert}
	\\[0,15cm]
	\leqslant
	{{\exp}\big({\beta\lVert\boldsymbol{\widehat{h}}\rVert_1}\big)
		\over Z^{\rm f}_{G;\beta,q,\boldsymbol{\widehat{0}}}} \,
	\sum_{\omega}f(\omega)
	\prod_{xy\in E} \big[\exp({q\beta J_{xy}})-1\big]^{\omega_{xy}}\,
	q^{\big\vert\big\{C(\omega):\, C(\omega)\cap\partial V=\emptyset \big\}\big\vert},
	\end{multline*}
	because
	\begin{align*}
	\exp	
	\Bigg(
	\beta\sum_{x\in V}h_{x,{\rm max}}
	\Bigg)
	\leqslant
	\exp	
	\Bigg(
	\beta\sum_{x\in \mathbb{Z}^d} \max_{m\in\{1,\ldots,q\}}\vert h_{x,m}\vert
	\Bigg)
	\leqslant
	{\exp}\big({\beta\lVert\boldsymbol{\widehat{h}}\rVert_1}\big).
	\end{align*}
	Therefore, $\phi^{\rm w}_{G;\beta,q,\boldsymbol{\widehat{h}}} (f)
	\leqslant
	{\exp}\big({\beta\lVert\boldsymbol{\widehat{h}}\rVert_1}\big) 
	\phi^{\rm f}_{G;\beta,q,\boldsymbol{\widehat{0}}} (f)$.
	Hence, the validity of lemma, with $\#={\rm w}$, follows
	by taking $C(\beta,\boldsymbol{\widehat{h}})
	={\exp}\big({\beta\lVert\boldsymbol{\widehat{h}}\rVert_1}\big)$.
\end{proof}

\section{Proof of Theorem \ref{phase-transition-existence}}
\label{Proof of Theorems}

In this section we apply the results obtained in the last section to prove the main result of  paper and to establish  some remarks concerning the behavior of the random-cluster model on the graph $\mathbb{G}= (\mathbb{Z}^d, \mathbb{E}^d)$. Throughout this section, $G=(V,E)$ is a finite subgraph of $\mathbb{G}$.

\subsection*{Proof of Item 1}
Let us define the critical point $\beta_c=\beta_c(q,\boldsymbol{\widehat{h}})$ by
\begin{align}\label{critical-prob}
\beta_c=
\inf\big\{\beta\geqslant 0: 
\theta(\beta,q,\boldsymbol{\widehat{h}})>0
\big\},
\end{align}
where $\theta$ is the percolation-probability defined in \eqref{percolation-probability}. Fix $x\in V$.
For each $n\geqslant 1$, let
$B(x,n)=\big\{y\in V: \Vert x-y\Vert\leqslant n\big\}$ be a box of radius $n$ with centre at $x$ and $\partial B(x,n)=\big\{y\in\mathbb{Z}^d : \Vert x-y\Vert= n\big\}$ its surface. 
Since 
$\big\{x\longleftrightarrow \partial B(x,n)\big\}$ 
is an increasing event, by combining 
Lemma \ref{Stochastic domination of Bernoulli} and Corollary \ref{ineq-fund}, it follows that
\begin{align*}
\phi^{\rm f}_{G;\beta,q,\boldsymbol{\widehat{0}}} 
\big[x\longleftrightarrow \partial B(x,n)\big]
\leqslant 
\phi^{\rm f}_{G;\beta,q,\boldsymbol{\widehat{h}}} 
\big[x\longleftrightarrow \partial B(x,n)\big]
\leqslant 
\mathbb{P}_{G;\beta} \big[x\longleftrightarrow \partial B(x,n)\big].
\end{align*}
Letting $G\uparrow\mathbb{G}$, from Theorem \ref{conv}-Item (a) we have
\begin{align}\label{ineq-inf-vol}
\phi^{\rm f}_{\mathbb{G};\beta,q,\boldsymbol{\widehat{0}}} 
\big[x\longleftrightarrow \partial B(x,n)\big]
\leqslant 
\phi^{\rm f}_{\mathbb{G};\beta,q,\boldsymbol{\widehat{h}}} 
\big[x\longleftrightarrow \partial B(x,n)\big]
\leqslant 
\mathbb{P}_{\mathbb{G};\beta} \big[x\longleftrightarrow \partial B(x,n)\big],
\end{align}
where $\mathbb{P}_{\mathbb{G};\beta}$ denotes the infinite-volume Bernoulli measure.
Since $\{x\longleftrightarrow \partial B(x,n)\} \downarrow \{\vert C_x\vert=\infty\}$ as $n\to \infty$, it follows from the continuity of the measure that
\begin{align*}
\phi^{\rm f}_{\mathbb{G};\beta,q,\boldsymbol{\widehat{0}}} \big( \vert C_x\vert=\infty\big)
\leqslant 
\phi^{\rm f}_{\mathbb{G};\beta,q,\boldsymbol{\widehat{h}}} \big( \vert C_x\vert=\infty\big)
\leqslant 
\mathbb{P}_{\mathbb{G};\beta} \big( \vert C_x\vert=\infty \big), \quad  x\in\mathbb{Z}^d.
\end{align*}
Then, by using the definition of function $\theta$,
\begin{align}\label{ineq-theta}
\theta(\beta,q,\boldsymbol{\widehat{0}})
\leqslant 
\theta(\beta,q,\boldsymbol{\widehat{h}})
\leqslant 
\theta(\beta,1,\boldsymbol{\widehat{0}}), \quad \beta\in[0,\infty].
\end{align}
By definition of $\beta_c$, this implies that
\begin{align}\label{ineq-cp}
\beta_c (q,\boldsymbol{\widehat{0}})
\geqslant 
\beta_c (q,\boldsymbol{\widehat{h}}) 
\geqslant 
\beta_c (1,\boldsymbol{\widehat{0}}).
\end{align}
Since $0< 1-\exp\big[{-q\beta_c (1,\boldsymbol{\widehat{0}})\, J}\big], 1-\exp\big[{-q\beta_c (q,\boldsymbol{\widehat{0}})\, J}\big] <1$ for $d\geqslant 2$ (see \cite{Grimmett99,Grimmett2}), we deduce the important fact that 
\[
0<\beta_c (q,\boldsymbol{\widehat{h}})<\infty, \quad  q\in\{2,3,\ldots\}.
\]
That is, the critical point $\beta_c$ is non-trivial.
On the other hand, it is known that $\theta(\beta,q,\boldsymbol{\widehat{0}})>0$ when $\beta>\beta_c (q,\boldsymbol{\widehat{0}})$, and that $\theta(\beta,1,\boldsymbol{\widehat{0}})=0$ when $\beta<\beta_c (\boldsymbol{\widehat{0}},1).$ Using this in \eqref{ineq-theta} we have
\[
\theta(\beta,q,\boldsymbol{\widehat{h}})
\left\{
\begin{array}{lllll}
=0 & \text{if} \ \beta<\beta_c (1,\boldsymbol{\widehat{0}}),
\\[0,1cm]
>0 & \text{if} \ \beta>\beta_c (q,\boldsymbol{\widehat{0}}).
\end{array}
\right.
\]
Since $\theta$ is an increasing function in $\beta$, see Proposition \ref{Monotonicity with respect to the coupling}-Item (b), 
the proof follows.

\qed

\begin{remark1}\label{phase-transition-existence-1}
	Consider the same notation as the proof of Item 1).
	For $\beta<\beta_0$, where $\beta_0>0$ is a certain real number less than $\beta_c$, given in \eqref{critical-prob}, 
	there exists $\psi(\beta)>0$ such that  
	\begin{align}\label{exp-tail-dec}
	\sup_{x\in\mathbb{Z}^d}
	\phi^{\rm f}_{\mathbb{G};\beta,q,\boldsymbol{\widehat{h}}} 
	\big[x\longleftrightarrow \partial B(x,n)\big]
	\leqslant 
	{\exp}\big[-n \psi(\beta)\big]
	\quad \text{for all} \ n.
	\end{align}
	In other words, we have the exponential tail decay of the radius of an open cluster for sufficiently small parameter values.
	Indeed, by \eqref{ineq-inf-vol}, we have
	\begin{align}\label{ineq-inf-vol-1}
	\phi^{\rm f}_{\mathbb{G};\beta,q,\boldsymbol{\widehat{h}}} 
	\big[x\longleftrightarrow \partial B(x,n)\big]
	\leqslant 
	\mathbb{P}_{\mathbb{G};\beta} \big[x\longleftrightarrow \partial B(x,n)\big], \quad  x\in\mathbb{Z}^d.
	\end{align}
	Moreover,
	for any $\beta<\beta_c (1,\boldsymbol{\widehat{0}})$ it is known that there exists $\psi(\beta)>0$ such that, see, e.g., \cite{Grimmett2,Sim},
	or Corollary 9.38 in \cite{Grimmett2}, 
	\begin{align}\label{ineq-final-1}
	\mathbb{P}_{\mathbb{G};\beta} \big[x\longleftrightarrow \partial B(x,n)\big]
	\leqslant 
	{\exp}\big[-n \psi(\beta)\big] \quad \text{for all} \ n.
	\end{align}
	By combining \eqref{ineq-inf-vol-1} and \eqref{ineq-final-1}, the inequality in \eqref{exp-tail-dec} follows by taking $\beta_0=\beta_c (1,\boldsymbol{\widehat{0}})$.
\end{remark1}

\begin{remark1}\label{unique-cluster-theorem}
	In the random-cluster model a percolation transition  is guaranteed, because 
	$\phi^{\rm w}_{\mathbb{G};\beta,q,\boldsymbol{\widehat{h}}} (N_{\infty}= 1)=0$ when $\beta<\beta_c,$  and  
	$ \phi^{\rm f}_{\mathbb{G};\beta,q,\boldsymbol{\widehat{h}}} (N_{\infty}= 1)=1$ when $\beta>\beta_c$.
	Indeed, from Theorem \ref{phase-transition-existence}-Item 1), for any $\beta>\beta_c$,
	\begin{equation*}\label{id-1}
	0<\theta(\beta,q,\boldsymbol{\widehat{h}})
	=
	\sup_{x\in\mathbb{Z}^d}
	\phi^{\rm f}_{\mathbb{G};\beta,q,\boldsymbol{\widehat{h}}} 
	\big(\vert C_x\vert=\infty\big)
	\leqslant 
	\phi^{\rm f}_{\mathbb{G};\beta,q,\boldsymbol{\widehat{h}}} 
	(N_{\infty}\geqslant 1).
	\end{equation*}
	Since $\{N_{\infty}\geqslant 1\}$ is a tail event and $\phi^{\rm f}_{\mathbb{G};\beta,q,\boldsymbol{\widehat{h}}}$ is an extremal measure (see Theorem \ref{conv}-Item (b)), it follows from the uniqueness of the infinite cluster (see Theorem \ref{unique-cluster}) that $\phi^{\rm f}_{\mathbb{G};\beta,q,\boldsymbol{\widehat{h}}} (N_{\infty}= 1)=1$.
	On the other hand, if $\beta<\beta_c$, by Theorem \ref{phase-transition-existence}-Item 1), we have  $\phi^{\rm f}_{\mathbb{G};\beta,q,\boldsymbol{\widehat{h}}} (\vert C_x\vert=\infty)=0$ for all $x\in\mathbb{Z}^d$. 
	Again, from the uniqueness of the infinite cluster, 
	\begin{align*}
	0=\phi^{\rm f}_{\mathbb{G};\beta,q,\boldsymbol{\widehat{h}}} (N_{\infty}>0)=\phi^{\rm f}_{\mathbb{G};\beta,q,\boldsymbol{\widehat{h}}} (N_{\infty}=1).
	\end{align*}
	Finally, since the infinite-volume measure is unique when $\beta<\beta_c$, that is, $\phi^{\rm f}_{\mathbb{G};\beta,q,\boldsymbol{\widehat{h}}}=\phi^{\rm w}_{\mathbb{G};\beta,q,\boldsymbol{\widehat{h}}}$ (see Theorem 11-Item (ii) in \cite{CV16}), we have $\phi^{\rm w}_{\mathbb{G};\beta,q,\boldsymbol{\widehat{h}}} (N_{\infty}=1)=0.$
\end{remark1}

\subsection*{Proof of Item 2}
We follow a similar argument as in reference \cite{Grimmett99} p. 46 and the same notation as the proof of Item 1).	
Let $M_x$ be the random variable defined by 
$M_x=\max\big\{n:\{x\longleftrightarrow \partial B(x,n)\} \ \text{occurs}\big\}$ for each
$x\in\mathbb{Z}^d$. If $\beta<\beta_c$, by Remark \ref{unique-cluster-theorem} we have  $\phi^{\rm w}_{\mathbb{G};\beta,q,\boldsymbol{\widehat{h}}} (M_x<\infty)=1$, that is, the collection $\{M_x=n\}_{n\geqslant 1}$ is a partition of the sample space $\{0,1\}^{\mathbb{E}^d}$. This gives
\begin{align*}
\phi^{\rm w}_{\mathbb{G};\beta,q,\boldsymbol{\widehat{h}}} 
\big(\vert C_x\vert\big)
&\leqslant
\sum_{n=1}^{\infty} 
\phi^{\rm w}_{\mathbb{G};\beta,q,\boldsymbol{\widehat{h}}} 
\big(\vert C_x\vert \big\vert M_x=n\big)
\phi^{\rm w}_{\mathbb{G};\beta,q,\boldsymbol{\widehat{h}}}
(M_x=n)
\\
&\leqslant
\sum_{n=1}^{\infty} \big\vert B(x,n)\big\vert \,
\phi^{\rm w}_{\mathbb{G};\beta,q,\boldsymbol{\widehat{h}}}
\big[x\longleftrightarrow \partial B(x,n)\big]
\quad \text{for all} \ x\in\mathbb{Z}^d.
\end{align*}
Since 
$\vert B(x;n)\vert\leqslant \pi(d)(n+1)^d$ for some constant $\pi(d)$, from  Remark  \ref{phase-transition-existence-1} we have, for $\beta<\beta_0$, 
\begin{align*}
\phi^{\rm w}_{\mathbb{G};\beta,q,\boldsymbol{\widehat{h}}} 
\big(\vert C_x\vert\big)
\leqslant \pi(d)
\sum_{n=1}^{\infty} (n+1)^d\, {\exp}\big[-n \psi(\beta)\big]<\infty
\quad \text{for all} \ x\in\mathbb{Z}^d.
\end{align*}
Taking supremum on $x$ in the above inequality, by definition \eqref{percolation-probability} of the susceptibility function $\chi$, we have that $\chi$ is finite when $\beta<\beta_0$.
On the other hand, note that
\begin{align*}
\phi^{\rm w}_{\mathbb{G};\beta,q,\boldsymbol{\widehat{h}}} 
\big(\vert C_x\vert\big)
&=
\infty \boldsymbol{\cdot} 
\phi^{\rm w}_{\mathbb{G};\beta,q,\boldsymbol{\widehat{h}}} 
\big(\vert C_x\vert=\infty\big)
+
\sum_{n=1}^\infty n \, 
\phi^{\rm w}_{\mathbb{G};\beta,q,\boldsymbol{\widehat{h}}} 
\big(\vert C_x\vert=n\big)
\\
&\geqslant 
\infty \boldsymbol{\cdot} 
\phi^{\rm f}_{\mathbb{G};\beta,q,\boldsymbol{\widehat{h}}} 
\big(\vert C_x\vert=\infty\big)
\quad \text{for all} \ x\in\mathbb{Z}^d.
\end{align*}
Taking supremum on $x$ in the above inequality, by definition \eqref{percolation-probability} of $\theta$ and $\chi$,
we have 
\begin{align*}
\chi(\beta,q,\boldsymbol{\widehat{h}})
\geqslant 
\infty \boldsymbol{\cdot} \theta(\beta,q,\boldsymbol{\widehat{h}}).
\end{align*}
By Theorem \ref{phase-transition-existence}-Item 1), $\theta(\beta,q,\boldsymbol{\widehat{h}})>0$  when $\beta>\beta_c$. Then,  it follows that
$\chi(\beta,q,\boldsymbol{\widehat{h}})=\infty$ when $\beta>\beta_c$. Thus the proof of the second item is complete.	

\qed

\begin{remark1}\label{unique-cluster-theorem-1} 
	There exists a critical point $0<\pi_c=\pi_c(q,\boldsymbol{\widehat{h}})<\infty$ such that
	$\chi(\beta,q,\boldsymbol{\widehat{h}})<\infty$ if $\beta<\pi_c,$ and $\chi(\beta,q,\boldsymbol{\widehat{h}})=\infty$ if $\beta>\pi_c.$
	In other words,  
	the susceptibility function $\chi$ experiences a phase transition phenomenon.
	Indeed, let us define 
	\begin{align*}
	\pi_c=
	\sup\big\{\beta\geqslant 0: 
	\chi(\beta,q,\boldsymbol{\widehat{h}})<\infty 
	\big\},
	\end{align*}
	where $\chi$ is as in \eqref{percolation-probability}. Note that $\pi_c\leqslant \beta_c$, where $\beta_c$ is the critical probability given in \eqref{critical-prob}. 
	Since the event 
	$\{x\longleftrightarrow y\}$ 
	is increasing, by combining
	Lemma \ref{Stochastic domination of Bernoulli} and Corollary \ref{ineq-fund}, in the thermodynamic limit, we have
	\begin{align*}
	\phi^{\rm w}_{\mathbb{G};\beta,q,\boldsymbol{\widehat{0}}} (x\longleftrightarrow y)
	\leqslant 
	\phi^{\rm w}_{\mathbb{G};\beta,q,\boldsymbol{\widehat{h}}} (x\longleftrightarrow y)
	\leqslant 
	\mathbb{P}_{\mathbb{G};\beta} (x\longleftrightarrow y)  \quad \text{for all} \ x,y\in\mathbb{Z}^d.
	\end{align*}
	Combining this with the identity
	\begin{align}\label{id-sus}
	\chi(\beta,q,\boldsymbol{\widehat{h}})= \sup_{x\in\mathbb{Z}^d}\sum_{y\in\mathbb{Z}^d}
	\phi^{\rm w}_{\mathbb{G};\beta,q,\boldsymbol{\widehat{h}}} (x\longleftrightarrow y),
	\end{align}
	we obtain
	$
	\chi(\beta,q,\boldsymbol{\widehat{0}})
	\leqslant
	\chi(\beta,q,\boldsymbol{\widehat{h}})
	\leqslant
	\chi(\beta,1,\boldsymbol{\widehat{0}}).
	$
	Consequently, by definition of $\pi_c$, we have
	\begin{align}\label{ineq-cp-1}
	\pi_c (1,\boldsymbol{\widehat{0}})
	\leqslant 
	\pi_c (q,\boldsymbol{\widehat{h}}) 
	\leqslant 
	\pi_c (q,\boldsymbol{\widehat{0}}).
	\end{align}
	It is well-known that $\pi_c (1,\boldsymbol{\widehat{0}})=\beta_c (1,\boldsymbol{\widehat{0}})$ (see \cite{Aiz87,Mensh86,Mensh86-1}) and that $\pi_c (q,\boldsymbol{\widehat{0}})\leqslant \beta_c (q,\boldsymbol{\widehat{0}})$.
	Since $0< 1-\exp\big[{-q\beta_c (1,\boldsymbol{\widehat{0}})\, J}\big], 1-\exp\big[{-q\beta_c (q,\boldsymbol{\widehat{0}})\, J}\big] <1$  for $d\geqslant 2$ (see \cite{Grimmett99,Grimmett2}), we have that
	$\pi_c (q,\boldsymbol{\widehat{h}})$ is non-trivial for all $q\in\{2,3,\ldots\}.$
	On the other hand,
	by Theorem \ref{phase-transition-existence}-Item 2), $\chi(\beta,q,\boldsymbol{\widehat{h}})=\infty$ when $\beta>\beta_c$, and  $\chi(\beta,q,\boldsymbol{\widehat{h}})<\infty$ when $\beta<\beta_0$.
	Since $\chi$ is an increasing function in $\beta$, see Proposition \ref{Monotonicity with respect to the coupling}-Item (b), the remark  follows.
\end{remark1}

\subsection*{Proof of Item 3}
From  Lemma \ref{lemma-teo2} with $f=\mathds{1}_{\{x\longleftrightarrow y\}}$, in the thermodynamic limit, we have
%
\begin{align*}
\phi^{\rm w}_{\mathbb{G};\beta,q,\boldsymbol{\widehat{h}}} (x\longleftrightarrow y)
\leqslant
C(\beta,\boldsymbol{\widehat{h}}) \,
\phi^{\rm f}_{\mathbb{G};\beta,q,\boldsymbol{\widehat{0}}} (x\longleftrightarrow y)
\quad \text{for all} \ x,y\in\mathbb{Z}^d,
\end{align*}
whenever $\boldsymbol{\widehat{h}}\in\ell^{1}(\mathbb{Z}^d)$ has positive terms.
It is known that, for $\beta<\beta_c (q,\boldsymbol{\widehat{0}})$ and $q\in\{2,3,\ldots\}$,
there exists $\gamma(\beta,q)>0$ such that (see main theorem in \cite{DRV19}), 
\[
\phi^{\rm f}_{\mathbb{G};\beta,q,\boldsymbol{\widehat{0}}}(x\longleftrightarrow y)
\leqslant 
{\exp}\big[{-\|x-y\| \gamma(\beta,q)} \big].
\]	
And thereby, for $\beta<\beta_c (q,\boldsymbol{\widehat{0}})$,
\begin{align}\label{des-dec}
\phi^{\rm w}_{\mathbb{G};\beta,q,\boldsymbol{\widehat{h}}} (x\longleftrightarrow y)
\leqslant
C(\beta,\boldsymbol{\widehat{h}}) \,
{\exp}\big[{-\|x-y\| \gamma(\beta,q)} \big] \quad \text{for all} \ x,y\in\mathbb{Z}^d.
\end{align}

In what follows we claim that
$\beta_c (q,\boldsymbol{\widehat{0}})
=
\beta_c (q,\boldsymbol{\widehat{h}}) $. Indeed, 
by Item \eqref{ineq-cp},
$
\beta_c (q,\boldsymbol{\widehat{0}})
\geqslant 
\beta_c (q,\boldsymbol{\widehat{h}}).
$	
Now, suppose that $\beta_c (q,\boldsymbol{\widehat{0}})
> 
\beta_c (q,\boldsymbol{\widehat{h}})$. From Theorem \ref{phase-transition-existence}-Item 2),
\begin{align}\label{interval}
\chi(\beta,q,\boldsymbol{\widehat{h}})=\infty \quad \text{when} \ \beta_c (q,\boldsymbol{\widehat{h}})
<\beta<
\beta_c(q,\boldsymbol{\widehat{0}}).
\end{align} 	
From \eqref{des-dec}, for $\beta<\beta_c (q,\boldsymbol{\widehat{0}})$ and \text{for all} $y\in\mathbb{Z}^d$,
\begin{align*}
\sum_{y\in\mathbb{Z}^d}
\phi^{\rm w}_{\mathbb{G};\beta,q,\boldsymbol{\widehat{h}}} (x\longleftrightarrow y)
\leqslant
C(\beta,\boldsymbol{\widehat{h}}) \,
\sum_{n=1}^\infty
{\exp}\big[{-n \gamma(\beta,q)} \big]
\Bigg(\sum_{y:\|x-y\|=n}1\Bigg)<\infty.
\end{align*}
By combining this with the identity
\eqref{id-sus}, we have $\chi(\beta,q,\boldsymbol{\widehat{h}})<\infty$ when $\beta<\beta_c (q,\boldsymbol{\widehat{0}})$, but this contradicts the statement in \eqref{interval}. Hence the claimed follows.

Therefore, the exponential decay \eqref{des-dec} is satisfied for all $\beta<\beta_c (q,\boldsymbol{\widehat{h}})$.
This completes the proof of the third item.	

\qed 

\begin{remark1}
	Considering the same notation as the proof of Item 1),
	notice that the third item
	of Theorem \ref{phase-transition-existence} could be replaced by the following statement.
	In the case that $\boldsymbol{\widehat{h}}\in\ell^{1}(\mathbb{Z}^d)$ has positive terms,
	for all $\beta> 0$
	there exist $C(\beta,\boldsymbol{\widehat{h}}),\gamma(\beta,q)>0$, such that
	\begin{align*}
	\sup_{x\in\mathbb{Z}^d}
	\phi^{\rm w}_{\mathbb{G};\beta,q,\boldsymbol{\widehat{h}}} 
	\big[x\longleftrightarrow \partial B(x,n)\big]
	\leqslant 
	C(\beta,\boldsymbol{\widehat{h}})\, 
	{\exp}\big[-n \gamma(\beta,q)\big]
	\quad \text{for all} \ n,
	\end{align*}
	whenever $\beta<\beta_c$.
\end{remark1}

\begin{remark1}\label{decay-exp}
	Let $x,z\in V$ be distinct vertices. A subset $W$ of $V$ is said to separate $x$ and $z$ if $x,z\notin W$ and every path from $x$ to $z$ contains some vertex of $W$.
	Let $q=2$ and let $\boldsymbol{\widehat{h}}\in\ell^{1}(\mathbb{Z}^d)$ be a summable positive external field.
	Assume that $x, z \in V$ are  distinct vertices,
	and that $W$ separate $x$ and $z$.
	For $\beta\geqslant 0$  there exists $C(\beta,\boldsymbol{\widehat{h}})>0$ such that 
	\begin{align}\label{ident-SL}
	\phi^{\rm f}_{G;\beta,q,\boldsymbol{\widehat{h}}} (x\longleftrightarrow z)
	\leqslant
	C(\beta,\boldsymbol{\widehat{h}})
	\sum_{y\in W} \phi^{\rm f}_{G;\beta,q,\boldsymbol{\widehat{h}}} (x\longleftrightarrow y) \, \phi^{\rm f}_{G;\beta,q,\boldsymbol{\widehat{h}}} (y\longleftrightarrow z).
	\end{align}
	The above correlation inequality is known as the (modified) Simon-Lieb inequality.
	Indeed,	
	since the event $\{x\longleftrightarrow y\}$ is monotone  increasing,
	Corollary \ref{ineq-fund} gives
	\begin{align}\label{ineq-fund-1}
	\phi^{\rm f}_{G;\beta,q,\boldsymbol{\widehat{0}}} (x\longleftrightarrow y)
	\leqslant 
	\phi^{\rm f}_{G;\beta,q,\boldsymbol{\widehat{h}}}(x\longleftrightarrow y)  \quad \text{for all} \ x,y\in V.
	\end{align}
	
	The Simon-Lieb inequality \cite{Grimmett2,H57,Lieb,Sim}
	for the measure $	\phi_{G;\beta,q,\boldsymbol{\widehat{0}}}$, with $q=2$, gives
	\begin{align*}
	\phi^{\rm f}_{G;\beta,q,\boldsymbol{\widehat{0}}} (x\longleftrightarrow z)
	&\leqslant 
	\sum_{y\in W} \phi^{\rm f}_{G;\beta,q,\boldsymbol{\widehat{0}}} (x\longleftrightarrow y)\, \phi^{\rm f}_{G;\beta,q,\boldsymbol{\widehat{0}}} (y\longleftrightarrow z)  
	\\
	&\stackrel{\eqref{ineq-fund-1}}{\leqslant}
	\sum_{y\in W} \phi^{\rm f}_{G;\beta,q,\boldsymbol{\widehat{h}}} (x\longleftrightarrow y)\, \phi^{\rm f}_{G;\beta,q,\boldsymbol{\widehat{h}}} (y\longleftrightarrow z).
	\end{align*}
	By combining this  with Lemma \ref{lemma-teo2}, the validity of the inequality in \eqref{ident-SL} follows.
\end{remark1}

\begin{remark1}
	Since $C(\beta,\boldsymbol{\widehat{h}})\geqslant 1$ and $\phi^{\rm f}_{\mathbb{G},\beta,q,\boldsymbol{\widehat{h}}}$ is a non-translation-invariant measure with $q=2$ and $\boldsymbol{\widehat{h}}\in\ell^{1}(\mathbb{Z}^d)$ has positive terms, note that
	we cannot use the Simon-Lieb inequality in \eqref{ident-SL}, as in the  translationally-invariant case, to prove
	the exponential decay of the connectivity function for two
	points.
\end{remark1}

\subsection*{Acknowledgments}
We would like to thank L. Cioletti for many valuable comments and careful reading of this manuscript.
This study was financed in part by the Coordena\c{c}\~ao de Aperfei\c{c}oamento de Pessoal de N\'ivel Superior - Brazil (CAPES) - Finance Code 001. 


\end{document}